\documentclass[aps,11pt,twoside, nofootinbib]{revtex4}
\usepackage{amsmath,amsthm,latexsym,amssymb,verbatim,enumerate,graphicx}
\usepackage[T1]{fontenc}

\usepackage{hyperref}
\usepackage{color}

\newtheorem*{theorem}{Theorem}
\newtheorem{corollary}{Corollary}
\newtheorem*{definition}{Definition}
\newtheorem{proposition}{Proposition}
\newtheorem{lemma}{Lemma}

\def\NORM#1{ {\left|\hspace{-.022in}\left| #1 \right|\hspace{-.022in}\right|} }

\def\Norm#1{ {\big|\hspace{-.022in}\big| #1 \big|\hspace{-.022in}\big|} }
\def\norm#1{ {|\hspace{-.022in}|#1|\hspace{-.022in}|} }

\def\Mfont#1{\mathsf{#1}}
\def\ALL{{\Mfont{ALL}}}
\def\LO{{\Mfont{LO}}}
\def\LOCC{{\Mfont{LOCC^{\rightarrow}}}}
\def\rightLOCC{{ \Mfont{LOCC}^\rightarrow }}
\def\leftLOCC{{ \Mfont{LOCC}^\leftarrow }}
\def\leftrightLOCC{{{\Mfont{LOCC}}}}
\def\Mclass{{\Mfont{M}}}
\def\SEP{{\Mfont{SEP}}}
\def\CS{{\mathcal{S}}}

\def\CC{{\mathcal{C}}}
\def\QMA{{\sf{QMA}}}

\def\BellQMA{{\sf{BellQMA}}}

\def\NP{{\sf{NP}}}

\def\SAT{{\sf{SAT}}}

\def\nn{\nonumber}
\def\norm#1{ {|\hspace{-.022in}|#1|\hspace{-.022in}|} }
\def\Norm#1{ {\big|\hspace{-.022in}\big| #1 \big|\hspace{-.022in}\big|} }

\def\NORM#1{ {\left|\hspace{-.022in}\left| #1 \right|\hspace{-.022in}\right|} }

\def\squareforqed{\hbox{\rlap{$\sqcap$}$\sqcup$}}
\def\qed{\ifmmode\squareforqed\else{\unskip\nobreak\hfil
\penalty50\hskip1em\null\nobreak\hfil\squareforqed
\parfillskip=0pt\finalhyphendemerits=0\endgraf}\fi}
\def\endenv{\ifmmode\;\else{\unskip\nobreak\hfil
\penalty50\hskip1em\null\nobreak\hfil\;
\parfillskip=0pt\finalhyphendemerits=0\endgraf}\fi}
\newenvironment{example}{\noindent \textbf{{Example~}}}{\qed}

\newcommand{\proj}[1]{|#1\rangle\langle #1|}
\newcommand{\bra}[1]{\langle #1|}
\newcommand{\ket}[1]{|#1\rangle}

\DeclareMathOperator{\tr}{tr}

\newcommand{\id}{\mathbb{I}}

\mathchardef\ordinarycolon\mathcode`\:
\mathcode`\:=\string"8000
\def\vcentcolon{\mathrel{\mathop\ordinarycolon}}
\begingroup \catcode`\:=\active
  \lowercase{\endgroup
  \let :\vcentcolon
  }

\newcommand{\nc}{\newcommand}
\nc{\rnc}{\renewcommand} \nc{\beq}{\begin{equation}}
\nc{\eeq}{{\end{equation}}} \nc{\bea}{\begin{eqnarray}}
\nc{\eea}{\end{eqnarray}} \nc{\beqa}{\begin{eqnarray}}
\nc{\eeqa}{\end{eqnarray}} \nc{\lbar}[1]{\overline{#1}}
\def\BRa#1{{\Big\langle #1 \Big|  }}
\def\KEt#1{{\Big| #1 \Big\rangle }}

\nc{\conv}{\operatorname{conv}}
\nc{\smfrac}[2]{\mbox{$\frac{#1}{#2}$}} \nc{\Tr}{\operatorname{Tr}}
\nc{\ox}{\otimes} \nc{\dg}{\dagger} \nc{\dn}{\downarrow}
\nc{\lmax}{\lambda_{\text{max}}}
\nc{\lmin}{\lambda_{\text{min}}}

\nc{\supp}{{\operatorname{supp }}}
\nc{\csupp}{{\operatorname{csupp}}}
\nc{\qsupp}{{\operatorname{qsupp}}} \nc{\var}{\operatorname{var}}
\nc{\rar}{\rightarrow} \nc{\lrar}{\longrightarrow}
\nc{\poly}{\operatorname{poly}}
\nc{\polylog}{\operatorname{polylog}} \nc{\Lip}{\operatorname{Lip}}
\nc{\mb}[1]{\mathbf{#1}}
\nc{\ep}{\epsilon}
\nc{\Om}{\Omega}
\nc{\wt}[1]{\widetilde{#1}}

\def\>{\rangle}
\def\<{\langle}

\def\om{\omega}
\def\sig{\sigma}

\nc{\glneq}{{\raisebox{0.6ex}{$>$}  \hspace*{-1.8ex} \raisebox{-0.6ex}{$<$}}}
\nc{\gleq}{{\raisebox{0.6ex}{$\geq$}\hspace*{-1.8ex} \raisebox{-0.6ex}{$\leq$}}}

\nc{\vholder}[1]{\rule{0pt}{#1}}
\nc{\wh}[1]{\widehat{#1}}
\nc{\h}[1]{\widehat{#1}}

\nc{\ob}[1]{#1}

\def\beq{\begin {equation}}
\def\eeq{\end {equation}}

\def\be{\begin{equation}}
\def\ee{\end{equation}}

\nc{\eq}[1]{Eq.~(\ref{eq:#1})} \nc{\eqs}[2]{Eqs.~(\ref{eq:#1}) and
(\ref{eq:#2})}

\nc{\eqn}[1]{Eq.~(\ref{eqn:#1})}
\nc{\eqns}[2]{Eqs.~(\ref{eqn:#1}) and (\ref{eqn:#2})}

\nc{\region}{\cS\cW}

\begin{document}

\title{{\Large Faithful Squashed Entanglement}}

\author{Fernando G.S.L. Brand\~ao}
\email{fgslbrandao@gmail.com}
\affiliation{Departamento de F\'isica, \\ 
Universidade Federal de Minas Gerais, \\
     Belo Horizonte, Caixa Postal 702, \\ 
     30123-970, MG, Brazil}
\author{Matthias Christandl}
\email{christandl@phys.ethz.ch}
\affiliation{Institute for Theoretical Physics, \\ 
ETH Zurich, \\ 
Wolfgang-Pauli-Strasse 27, CH-8057 Zurich, Switzerland}
\author{Jon Yard}
\email{jtyard@lanl.gov}
\affiliation{Center for Nonlinear Studies (CNLS) \\
Computer, Computational and Statistical Sciences (CCS-3) \\
Los Alamos National Laboratory \\
Los Alamos, NM 87545}

%\date{18 October 2006}

\begin{abstract}

Squashed entanglement is a measure for the entanglement of bipartite quantum states. In this paper we present a lower bound for squashed entanglement in terms of a distance to the set of separable states. This implies that squashed entanglement is faithful, that is, it is strictly positive if and only if the state is entangled.

We derive the lower bound on squashed entanglement from a lower bound on the quantum conditional mutual information which is used to define squashed entanglement. The quantum conditional mutual information corresponds to the amount by which strong subadditivity of von Neumann entropy fails to be saturated. Our result therefore sheds light on the structure of states that \emph{almost} satisfy strong subadditivity with equality. The proof is based on two recent results from quantum information theory:  the operational interpretation of the quantum mutual information as the optimal rate for state redistribution and the interpretation of the regularised relative entropy of entanglement as an error exponent in hypothesis testing. 

The distance to the set of separable states is measured in terms of the one-way LOCC norm, an operationally motivated norm giving the optimal probability of distinguishing two bipartite quantum states, each shared by two parties, using any protocol formed by local quantum operations and one-directional classical communication between the parties. A similar result for the Frobenius or Euclidean norm follows as an immediate consequence. 

The result has two applications in complexity theory. The first application is a quasipolynomial-time algorithm solving the weak membership problem for the set of separable states in one-way LOCC or Euclidean norm. The second application concerns quantum Merlin-Arthur games. Here we show that multiple provers are not more powerful than a single prover when the verifier is restricted to one-way LOCC operations thereby providing a new characterisation of the complexity class $\QMA$.
%This answers a question posed by Aaronson \emph{et al.} (Theory of Computing {\bf 5}, 1 (2009)) and 
\end{abstract}
\maketitle
\newpage
\tableofcontents
\parskip .75ex

%%%%%%%%%%%%%%%%%%%%%%%%%%%%%%%%%%%%%%%%%%%%%%%%%%%%%%%%%%%%%%%%%%%%%%%%

\section{Introduction}

The correlations of a bipartite quantum state $\rho_{AB}$ can be measured by the quantum mutual information
\begin{equation}
I(A;B)_{\rho} := H(A)_{\rho} + H(B)_{\rho} - H(AB)_{\rho},
\end{equation}
where $H(X)_{\rho} := -  \tr(\rho_X \log \rho_X)$ is the von Neumann entropy. Subadditivity of entropy implies that mutual information is always positive.  If $\Vert * \Vert_1$ is the trace norm,  defined as $\Vert X \Vert_1=\tr \sqrt{X^\dagger X}$, the inequality
\begin{equation} \label{mipinsker}
I(A;B)_{\rho} \geq \frac{1}{2 \ln 2} \Norm{\rho_{AB} - \rho_{A}\otimes \rho_{B}}^{2}_1,
\end{equation}
follows from Pinsker's inequality \cite{OP04} for the relative entropy.  Thus a bipartite state has zero mutual information if --  and only if --  it has no correlations (i.e.\ it is a product state).

Conditional mutual information, in turn, measures the correlations of two quantum systems relative to a third. Given a tripartite state $\rho_{ABE}$, it is defined as
\begin{equation}
I(A;B|E)_{\rho} := H(AE)_{\rho} + H(BE)_{\rho} - H(ABE)_{\rho} - H(E)_{\rho}.
\end{equation}
This measure is also always non-negative, as a consequence of the celebrated strong subadditivity inequality for the von Neumann entropy proved by Lieb and Ruskai \cite{LR73}. 

As for mutual information, one can ask which states have zero conditional mutual information.  Such a characterization was given in Ref.~\cite{HJPW04}.  A state $\rho_{ABE}$ has $I(A;B|E)_\rho = 0$ if and only if it is a so-called quantum Markov chain, i.e.\ there is a decomposition of the $E$ system vector space 
\begin{equation}
{\cal H}_E = \bigoplus_{j}  {\cal H}_{e_j^{L}} \otimes  {\cal H}_{e_j^{R}}
\end{equation}
into a direct sum of tensor products such that
\begin{equation}\label{eq:markov}
\rho_{ABE} = \sum_{j} p_j \rho_{Ae_j^{L}}\otimes \rho_{Be_j^{R}},
\end{equation}
with states $\rho_{Ae_j^{L}}$ on ${\cal H}_{A} \otimes  {\cal H}_{e_j^{L}}$, $\rho_{Be_j^{R}}$ on ${\cal H}_{B} \otimes  {\cal H}_{e_j^{R}}$ and probabilities $p_j$. One might ask: Is there an inequality analogous to \eqref{mipinsker} for the conditional mutual information? Here one could expect the minimum distance to quantum Markov chain states to play the role of the distance to the tensor product of the reductions in \eqref{mipinsker}. Up to now, however, only negative results were obtained in this direction \cite{INW08}. 

Note, in particular, that the $AB$ reduction of any state of the form~\eqref{eq:markov} is separable \cite{Wer89a}, i.e.\
\begin{equation}
\rho_{AB} = \sum_j p_j \rho_{A, j} \otimes \rho_{B, j}.
\end{equation}
Rather than finding a lower bound in terms of a distance to Markov chain states, our main result, the theorem in Section~\ref{sec:results}, is a lower bound, similar to the one in \eqref{mipinsker}, in terms of a distance to separable states. 

The fact that states $\rho_{ABE}$ with zero conditional mutual information are such that $\rho_{A:B}$ is separable motivated the introduction of the entanglement measure \textit{squashed entanglement}~\cite{CW04},\footnote{The functional (without the factor $\frac{1}{2}$) has also been considered in~\cite{Tucci1, Tucci2}.} defined as
\begin{equation}
E_{sq}(\rho_{A:B}) = \inf \left \{ \frac{1}{2} I(A;B|E)_\rho : \hspace{0.1 cm} \rho_{ABE} \hspace{0.1 cm} \text{is an extension of} \hspace{0.1 cm} \rho_{AB}            \right \}.
\end{equation}
This measure is known to be an upper bound both for distillable entanglement and distillable key \cite{CW04, Chr06, CSW08}. It satisfies several useful properties such as monotonicity under local operations and classical communication (LOCC), additivity on tensor products \cite{CW04}, and monogamy \cite{KW04}.% (embodied in the Coffman-Kundu-Wootters inequality \cite{CKW00}). 

A central open question in entanglement theory, posed already in \cite{CW04}, is whether squashed entanglement is a faithful entanglement measure, i.e.\ whether it vanishes if and only if the state is separable. On the one hand, an entangled state with zero squashed entanglement would imply the existence of \textit{bound key},~ i.e.\ of an entangled state from which no secret key can be extracted.  Analogously, a state with a non-positive partial transpose (NPT) \cite{Per96, HHH96} and vanishing squashed entanglement would show there are NPT-bound entangled states \cite{HHH98, HHHH09}. Both of these are core unresolved problems in quantum information theory. On the other hand, if squashed entanglement turned out to be positive on every entangled state, then we would have the first example of an entanglement measure which is faithful, LOCC monotone and which satisfies the monogamy inequality \cite{CKW00}, three properties usually attributed to entanglement but which are not known to be compatible with each other.

The lower bound on conditional mutual information obtained in this paper resolves this question by showing that squashed entanglement is strictly positive on every entangled state. Besides its impact on entanglement theory, this result has a number of unexpected consequences for separability testing, quantum data hiding and quantum complexity theory. In the following section we will present the results of this work in detail, before proceeding to the proofs in subsequent sections.

\section{Results}
\label{sec:results}
\vspace{0.2 cm}
\noindent \textbf{A lower bound on conditional mutual information:}  Our main result is an \textit{approximate} version of the fact that states $\rho_{ABE}$ with zero conditional mutual information $I(A;B|E)$ are such that $\rho_{A:B}$ is separable: We show that if a tripartite state has small conditional mutual information, its $AB$ reduction is \textit{close} to a separable state. 

The appropriate distance measure on quantum states turns out to be crucial for a fruitful formulation of the result and its consequences. By analogy with the definition of the trace norm as the optimal probability of distinguishing two quantum states, 
our main results involve norms that quantify the distinguishability of quantum states under measurements that are restricted by \emph{locality}.  Specifically, given two bipartite equiprobable states $\rho_{A:B}$ and $\sigma_{A:B}$, define the distance\footnote{This definition extends to a norm for the operators on $AB$, see~\eqref{eq:norm}.} 
$\norm{\rho - \sig}_{\LOCC} = 4\big(P_\mathrm{succ} - \frac 12\big)$, where $P_\mathrm{succ}\geq \frac 12$ is the probability of correctly distinguishing the states, maximised over all decision  procedures that can be implemented using only local operations and one-way classical communication, with forward classical communication from $A$ to $B$ \cite{MWW09}.  When this distance is small, the states $\rho$ and $\sigma$ will roughly behave ``in the same way" when used in any one-way LOCC protocol. 
Writing the distance from a state $\rho_{AB}$ to the set $\CS_{A:B}$ of separable states on $A\!:\!B$ as 
\begin{align} \norm{\rho_{AB} - \CS_{A:B}} = \min_{\sig \in \CS_{A:B}} \norm{\rho_{AB} - \sig},\end{align}
we may state our main result. \footnote{In a previous version of this paper a similar bound on the conditional mutual information had been claimed in terms of the LOCC distance, with no restriction on the classical communication used. However the statement was based on a flawed recursion argument for going from the one-way LOCC norm to the full LOCC norm. Thus we leave the question of whether one can strengthen the bound to the LOCC norm as an open question.}  
\begin{theorem} 
For every tripartite finite-dimensional state $\rho_{ABE}$, 
\begin{equation} \label{boundonewayLOCCnorm}
I(A;B|E)_{\rho} \geq \frac{1}{8 \ln 2}\Norm{\rho_{AB} - \CS_{A:B}}_{\LOCC}^{2},
\end{equation}
\end{theorem}

Interestingly, a similar statement is true when replacing the LOCC norm with the Frobenius norm, since there is the dimension-independent lower bound~\cite{MWW09}
\begin{equation} \label{relation2norm}
\norm{X}_\LOCC \geq \frac{1}{\sqrt{153}} \norm{X}_2,
\end{equation}
where $\norm{X}_2=\sqrt{\tr X^\dagger X}$. The Frobenius norm  is also known as Euclidean norm since it measures the Euclidean distance between quantum states when interpreted as elements of the real vector space of Hermitian matrices.
Note, however, that the theorem fails dramatically if we replace the $\LOCC$ norm by the trace norm, a counterexample being a tripartite extension of the antisymmetric state\footnote{The antisymmetric state is a particular Werner state defined as $\omega_{AB} := \frac{(\id - \mathbb{F})}{d(d-1)}$, with $\mathbb{F}$ the flip, or swap, operator.} constructed in \cite{CSW08}.\footnote{Indeed, Ref. \cite{CSW08} presents an extension $\omega_{ABE}$ of $\omega_{AB}$ such that $I(A;B|E)_{\omega} \leq \frac{4 \log e}{|A|-1}$, while a simple calculation gives that $\omega_{AB}$ is $\frac{1}{2}$ away from any separable state in trace-distance.} On the other hand, the theorem readily implies a lower bound for the trace norm with a dimension-dependent factor. This is because of~\eqref{relation2norm} and the well-known relation between the trace norm and the Frobenius norm
\begin{equation} \label{relation1norm}
\norm{X}_2 \geq \frac{1}{\sqrt{|A||B|}} \norm{X}_1,
\end{equation}
where $|A|$ denotes the dimension of system $A$ and $\norm{X}_1=\tr \sqrt{X^\dagger X}$.

\vspace{0.2 cm}
\noindent \textbf{A lower bound for squashed entanglement:}  Combining the Theorem with the definition of squashed entanglement we find the following corollary.
\begin{corollary} \label{corsquahsedlower}
For every state $\rho_{AB}$,
\begin{equation} \label{lowersquash}
E_{\rm sq}(\rho_{A:B}) \geq \frac{1}{16 \ln 2} \Norm{\rho_{AB} - \CS_{A:B}}_\LOCC^{2}.
\end{equation}
\end{corollary}
Because $\norm{*}_{\LOCC}$ is a norm, this implies that squashed entanglement is faithful, i.e.\ that it is strictly positive on every entangled state. This property had been conjectured in \cite{CW04} and its resolution here settles the last major open question regarding squashed entanglement. Squashed entanglement is the entanglement measure which -- among all known entanglement measures -- satisfies most properties that have been proposed as useful for an entanglement measure. A comparison of different entanglement measures is provided in Table~\ref{table-properties}.

\begin{table}
%\centering
%\begin{tiny}
\begin{tabular}{@{}l||l|l|l|l|l|l|l|l|l|l|l@{}}
 Measure       		& $E_{sq}$ \cite{CW04}	& $E_D$ \cite{BDSW96, Rains99}	&  $K_D$ \cite{DevWin05, HHHO05a}	&  $E_C$ \cite{HaHoTe01, BDSW96}	& $E_F$ \cite{BDSW96}	& $E_R$ \cite{VPRK97}	& $E_R^\infty$ \cite{VP98}	& $E_N$\cite{VidWer01}  \\
\hline \hline
normalisation        	&  y         		& y      	& y       	& y         	& y         	& y 		& y      		& y      	 \\
\hline
faithfulness     		&  y Cor.~\ref{corsquahsedlower}			& n \cite{HHH98}& ?	& y \cite{Yang05} & y	& y		& y \cite{BP10}		& n      	 \\
\hline
LOCC monotonicity\footnote{More precisely, we consider ``weak'' LOCC monotonicity, i.e.\ monotonicity without averages.}  &  y         		& y       	& y        	& y 		& y         	& y            & y              	& y \cite{P05}     	 \\
\hline
asymptotic continuity  &  y~\cite{AliFan04} & ?    	& ?        	& ?        	& y         	& y    \cite{DH99}    	& y \cite{Chr06}& n\cite{Chr06}      	 \\
\hline
convexity        		&  y         		& ?       	& ?  		& ?        	& y   		&  y 		&  y \cite{DoHoRu02}	& n      	 \\
\hline
strong superadditivity &  y         		& y       	& y        	& ? 		& n \cite{Shor03, H09}& n \cite{VW01}& ?	& ?    	 \\
\hline 
subadditivity		&  y   & ?  	& ?        	& y   		& y    	& y 		&  y 			& y    	 \\
\hline
monogamy    		&  y \cite{KW04}	& ? 		& ?   		& n \cite{CSW08} & n \cite{CSW08} 		& n \cite{CSW08}  		& n \cite{CSW08}   			& ? 		 \\
\end{tabular}
%\end{tiny}
\vspace{0.3cm} \caption{If no citation
is given, the property either follows directly from the definition or was derived by the authors of the main reference. Many recent results listed in this table have significance beyond the study of entanglement measures, such as Hastings's counterexample to the additivity conjecture of the minimum output entropy~\cite{H09} which implies that entanglement of formation is not strongly superadditive~\cite{Shor03}.} \label{table-properties}
\end{table}

Squashed entanglement is the quantum analogue of the intrinsic information, which is defined as
\begin{align}I(X ; Y\!\downarrow \!Z) := \inf_{P_{\bar{Z}|Z}} I(X;Y|\bar{Z}),\end{align}
for a triple of random variables $X, Y, Z$ \cite{Maurer1999}. The minimisation extends over all conditional probability distributions mapping $Z$ to $\bar{Z}$. It has been shown that the minimisation can be restricted to random variables $\bar{Z}$ with size $|\bar{Z}|= |Z|$\cite{Christandl2003}. This implies that the minimum is achieved and in particular that the intrinsic information only vanishes if there exists a channel $Z \rightarrow \bar{Z}$ such that $I(X;Y|\bar{Z})=0$. Whereas our work does not allow us to derive a dimension bound on the system $E$ in the minimisation of squashed entanglement and hence conclude that the minimisation is achieved in general, we can assert such a bound if squashed entanglement vanishes:  Corollary \ref{corsquahsedlower} implies that squashed entanglement vanishes only for separable $\rho_{AB}$. By Caratheodory's theorem, the number of terms in the separable decomposition of $\sum_i p_i \rho_{A,i} \otimes \rho_{B,i}$ can be bounded by $|AB|^2$, and thus $\rho_{ABE}=\sum_i p_i \rho_{A,i} \otimes \rho_{B,i}\otimes \ket{i}\bra{i}_E$ has vanishing conditional mutual information with $E=|AB|^2$. Equivalently, there exists a channel applied to a purification of $\rho_{AB}$ resulting in $\rho_{ABE}$ such that $I(A;B|E)_\rho$ vanishes.  

\vspace{0.2 cm}
\noindent \textbf{Positive rate in state redistribution:} Quantum conditional mutual information can be given an operational interpretation in the context of the \textit{state redistribution} problem \cite{DY07,YD07}. Suppose a sender and a receiver share a quantum state $\ket{\psi}_{ABEE'}$ and the sender (who initially holds subsystems $BE$) wants to send $B$ to the receiver (who initially holds $E'$), while preserving the purity of the global state. The purifying subsystem $A$ is unavailable to both of them. Given free entanglement between the sender and receiver, the minimum achievable rate of quantum communication needed to send $B$, with vanishing error in the limit of a large number of copies of the state, is given by $\frac{1}{2} I(A;B|E)_{\psi}$ \cite{DY07,YD07}. 
Then from an optimal protocol for state redistribution we find $E_{\rm sq}(\rho_{A:B})$ to be the minimum rate of quantum communication needed to send the $B$ system, optimised over all possible ways of distributing the side-information among $E$ and $E'$ \cite{Opp08}.

One can ask: Is a positive rate of quantum communication always required in order to redistribute a system entangled with another (irrespective of the side information available to the sender and receiver)? Corollary \ref{corsquahsedlower} allows us to answer this question in the affirmative. The need of a positive rate in state redistribution can then be seen as a new distinctive feature of quantum correlations.

\vspace{0.2 cm}
\noindent \textbf{The quantum de Finetti theorem, the LOCC norm and data-hiding states:} We say a bipartite state $\rho_{A:B}$ is $k$-extendible if there is a state $\rho_{A:B_1,...,B_k}$ that is permutation-symmetric in the $B$ systems with $\rho_{A:B} = \tr_{B_2,...,B_k}(\rho_{A:B_1,...,B_k}) $. Such a family of states provides a sequence of approximations to the set of separable states. A state is separable if, and only if, it is $k$-extendible for every $k$ \cite{HudMoo76, stormer, RW89, Wer89b, KR05, CKMR07}.

Indeed, quantum versions of the de Finetti theorem \cite{KR05, CKMR07} show that any $k$-extendible state $\rho_{A:B}$ is such that
\begin{equation}
\Norm{ \rho_{AB} - \CS_{A:B}}_{1} \leq \frac{4|B|^{2}}{k}.
\end{equation}
Moreover, this bound is close to tight, as there are $k$-extendible states that are $\Omega(|B|k^{-1})$-away from the set of separable states~\cite{CKMR07}.

Unfortunately, for many applications this error estimate -- exponentially large in the number of qubits of the state -- is too big to be useful. Our next result shows that a significant improvement is possible if we are willing to relax our notion of distance of two quantum states. It shows that in one-way LOCC norm we can obtain an error term that grows as the square root of the number of qubits of the $A$ system:
\begin{corollary} \label{monogamy}
Let $\rho_{A:B}$ be $k$-extendible. Then 
\begin{equation}
\Norm{\rho_{AB} - \CS_{A:B}}_{\LOCC} \leq  \sqrt{\frac{16 \ln 2 \log |A|}{k}}.
\end{equation}
\end{corollary}

The key point in the proof of Corollary \ref{monogamy} is the fact that squashed entanglement satisfies the so-called monogamy inequality \cite{KW04}, namely
\begin{equation} \label{monogamyinequality}
E_{\rm sq}(\rho_{A:B_1B_2}) \geq E_{\rm sq}(\rho_{A:B_1}) + E_{\rm sq}(\rho_{A:B_2}),
\end{equation}
for every state $\rho_{AB_1B_2}$. This, together with the bound $E_{\rm sq}(\rho_{A:B}) \leq \log |A|$ \cite{CW04} implies that the squashed entanglement of any $k$-extendible state must be smaller than $k^{-1}\log |A|$, which combined with Corollary \ref{corsquahsedlower} gives Corollary \ref{monogamy}. 

We do not know if the bound given in Corollary \ref{monogamy} is tight. An indication, however, is given by the following example, which shows that for all $k$ there is a $k$-extendible state $\rho_{AB}$ with $\log |A|=k$ and $\Norm{\rho_{AB} - \CS_{A:B}}_{\LOCC} \geq 1$. 

\begin{example}[Lower bound]
Consider systems $A=A_1A_2\cdots A_k$ and $B=B'B''$ and define 
\begin{equation}
\rho_{AB} := \tr_{B_2\cdots B_k }(\rho_{AB_1B_2\cdots B_k}) 
\end{equation}
where
\begin{equation}
\rho_{A_1...A_k:B'_1 B''_1 \cdots  B'_k B''_k} := \text{Id}_{A_1 \cdots A_k} \otimes \text{Sym}_{B'_1 B''_1, \ldots ,B'_k B''_k} \left( \Phi_{A_1B'_1} \otimes \proj{1}_{B''_1} \otimes \cdots \otimes  \Phi_{A_kB'_k} \otimes \proj{k}_{B''_k} \right),
\end{equation}
where $\Phi := \ket{\Phi}\bra{\Phi}$ is the projector onto an EPR pair $\ket{\Phi} := (\ket{0}\ket{0} + \ket{1}\ket{1})/\sqrt{2}$, $\text{Id}$ is the identity superoperator and $\text{Sym}_{B_1, \ldots , B_k}$ is the symmetrization superoperator defined as
\begin{equation}
\text{Sym}_{B_1, \ldots ,B_k}(X) := \frac{1}{k!} \sum_{\pi \in S_k} P_{\pi} X P_{\pi^{-1}},
\end{equation}
with $S_k$ the symmetric group of order $k$ and $P_{\pi}$ the representation of the permutation $\pi$ which acts on a $k$-partite vector space as $P_{\pi} \ket{l_1}\otimes ... \otimes \ket{l_{k}} = \ket{l_{\pi^{-1}(1)}} \otimes ... \otimes \ket{l_{\pi^{-1}(k)}}$. The state $\rho_{AB}$ takes the form
\begin{equation}
\rho_{AB} = \frac{1}{k} \sum_{l=1}^{k} \Phi_{A_lB'} \otimes \proj{l}_{B''} \otimes \tau_{A_1 \cdots  A_{l-1} A_{l+1} \cdots  A_k},
\end{equation}
with $\tau_{C} := \id/|C|$ by construction $k$-extendible. Moreover we can deterministically obtain an EPR pair from $\rho_{AB}$ by the following simple one-way LOCC protocol: Bob measures the $B''$ register in the computational basis and tells Alice the outcome $l$ obtained. Then she traces out all her systems except the $l$-th. Thus the one-way LOCC distance of $\rho_{AB}$ to separable states is at least the one-way LOCC distance of an EPR pair to separable states, which is known to be one~\cite{Virmani}, and so
\begin{equation} 
\Norm{\rho_{AB} - \CS_{A:B}}_{\leftLOCC} \geq \Omega(1).
\end{equation}
\end{example}
A direct implication of Corollary \ref{monogamy} concerns data-hiding states \cite{DLT02, DHT03, EW02, HLW06}. Every state $\rho_{AB}$ that can be well-distinguished from separable states by a global measurement, yet is almost completely indistinguishable from a separable state by LOCC measurements is a so-called data-hiding state: it can be used to hide a bit of information (of whether the state prepared is $\rho_{AB}$ or its closest separable state in LOCC norm) that is not accessible by LOCC operations alone. The antisymmetric state of sufficiently high dimension is an example of a data hiding state \cite{EW02}, as are random mixed states with high probability \cite{HLW06} (given an appropriate choice of the dimensions and the rank of the state). Corollary \ref{monogamy} shows that highly extendible states that are far away in trace norm from the set of separable states must necessarily be data-hiding (at least under one-way LOCC measurements).

\vspace{0.2 cm}
\noindent \textbf{A quasipolynomial-time algorithm for separability of quantum states:} Given a density matrix describing a bipartite quantum state $\rho_{A:B}$, the separability problem consists of determining if the state is entangled or separable. This is one of the most basic questions in quantum information theory and has been a topic of active research in the past years (see e.g.\ \cite{HHHH09}).

In the weak-membership problem for separability one should decide if a given bipartite state $\rho_{A:B}$ is in the interior of the set of separable states or $\epsilon$-away from any separable state (in trace norm), given the promise that one of the two alternatives holds true. The best known algorithms for the problem have worst case complexity $2^{\poly(|A|, |B|)\log(1/\epsilon)}$ (see e.g.\ \cite{DPS04, BV04, Ioa07, NOP09}). In fact, the problem is $\NP$-hard for $\epsilon = 1/\poly(n)$ \cite{Gur03, Gha10, Bei08}. 

%and, conditioned on the stronger assumption that there is no subexponential-time algorithm for 3-$\SAT$ \cite{IP01}, even subexponential algorithms of complexity up to $\exp ( O \left(\log^{1 - \nu} |A| \log^{1 - \mu} |B|\right))$ for \textit{constant} $\epsilon$  and any $\nu + \mu > 0$ can be ruled out \cite{HM10}.

%More precisely, let $\CC \subseteq \mathbb{R}^{d}$ be a convex set, $\Vert * \Vert$ a norm in $\mathbb{R}^{d}$.
%The weak-membership problem $W_{\epsilon, \norm{*}}(\CC)$ for $\CC$ is to decide if $x \in \CC$ or if $\norm{x - \CC} \geq \ep$ given the promise that one is the case. This formulation is fruitful because by disregarding points close to the boundary of $K$, one can safely ignore errors related the the finite precision of the input and the execution of the algorithm. 

Let ${\cal S}_{A:B}$ be the convex set of separable states acting on $A \otimes B$. Corollary \ref{monogamy} implies that for the one-way LOCC-norm version of the weak-membership problem for separability, one can greatly improve upon the previous known algorithms. 
\begin{corollary} \label{complexity}
There is a quasipolynomial-time algorithm for solving the weak membership problem for separability in the one-way LOCC norm and the Euclidean norm. More precisely, there is an $\exp ( O(\epsilon^{-2} \log|A| \log|B|))$-time algorithm for deciding $W_{\rm{SEP}}(\epsilon, \Vert * \Vert_{\Mfont{LOCC}})$ and $W_{\rm{SEP}}(\epsilon, \Vert * \Vert_2)$.
\end{corollary}

%We note that in many applications of the separability algorithm, e.g.\ accessing the usefulness of a quantum state for violating Bell's inequalities or for performing quantum teleportation, the one-way LOCC norm is actually the relevant quantity to consider. Also, 
It is intriguing that the complexity of our algorithm matches the best hardness result for the trace-norm version of the problem \cite{HM10}. It is an open question if a similar hardness result could be obtained for the one-way LOCC norm case, which would imply that our algorithm is essentially optimal.

The algorithm, which we analyse in more detail in Section \ref{twocorollaries}, is very simple and in fact is the basis of a well-known hierarchy of efficient tests for separability \cite{DPS04}. Using semidefinite-programming one looks for a $\Omega(\log|A| \epsilon^{-2})$-extension of $\rho_{A:B}$ and decides that the state is separable if one is found. If the state is separable, then clearly there is such an extension. If it is $\epsilon$-away from any separable state, Corollary \ref{monogamy} shows that no such extension exists.\footnote{In \cite{DPS04} one imposed the further constraint that the symmetric-extension has to have a positive partial transpose with respect to all $B$ systems. It is an interesting open question if the worst-case complexity of the algorithm (in the one-way LOCC-norm case) can be further improved taking into consideration this extra constraint.} 

More detailed information on this algorithm as well as the related problem of optimising a linear function over separable states, relevant for instance in mean-field theory, is given in~\cite{BCY11}.

\vspace{0.2 cm}
\noindent \textbf{Quantum Merlin-Arthur games with multiple Merlins:} A final application of the Theorem concerns quantum Merlin-Arthur games. The class $\QMA$ can be considered the quantum analogue of $\NP$ and is formed by all languages that can be decided in quantum polynomial-time by a verifier who is given a quantum system of polynomially many qubits as a proof (see e.g.\ \cite{Wat08}). It is natural to ask how robust this definition is. A few results are known in this direction: For example, it is possible to amplify the soundness and completeness parameters to exponential accuracy, even without enlarging the proof size \cite{MW05}. Also, the class does not change if we allow a first round of logarithmic-sized quantum communication from the verifier to the prover \cite{BSW10}. 

From Corollary~\ref{monogamy} we get a new characterisation of $\QMA$, which at first sight might appear to be strictly more powerful: We show $\QMA$ to be equal to the class of languages that can be decided in polynomial time by a verifier who is given a constant number $k$ of unentangled proofs and can measure them using any quantum polynomial-time implementable one-way LOCC protocol (among the $k$ proofs). We hope this characterisation of $\QMA$ shows useful in devising new $\QMA$ verifying protocols.

In order to formalise our result we consider the classes $\QMA_{\Mclass}(k)_{m,s,c}$, defined in analogy to $\QMA$ as follows \cite{KMY03, ABDFS08, HM10}:
\begin{definition}
A language $L$ is in $\QMA_{\Mclass}(k)_{m, s, c}$ if for every input $x \in \{ 0, 1\}^{n}$ there exists a poly-time implementable two outcome measurement $\{ M_x, \id - M_x \}$ from the class $\Mclass$ such that
\begin{itemize}
\item \textbf{Completeness:} If $x \in L$, there exist $k$ witnesses $\ket{\psi_1},...,\ket{\psi_k}$, each of $m$ qubits, such that
\begin{equation}
\tr \left( A_x \left(\ket{\psi_1}\bra{\psi_1} \otimes ... \otimes \ket{\psi_k}\bra{\psi_k} \right)  \right) \geq c
\end{equation}
\item \textbf{Soundness:} If $x \notin L$, then for any states $\ket{\psi_1}, ..., \ket{\psi_k}$
\begin{equation}
\tr \left( A_x \left(\ket{\psi_1}\bra{\psi_1} \otimes ... \otimes \ket{\psi_k}\bra{\psi_k} \right)  \right) \leq s
\end{equation}
\end{itemize}
We call $\QMA_{\Mclass}(k) = \QMA_{\Mclass}(k)_{\poly(n), 2/3, 1/3}$.
\end{definition}

Let $\Mclass = \SEP_{\text{YES}}$ be the class of (non-normalised) separable POVM elements.\footnote{POVM elements $A$ from this class form two-outcome POVMs $\{ A, \id - A \}$ known as separable-once-remove, meaning that they are separable measurements once one of the effects is removed. Here we use the notation $\SEP_{\text{YES}}$ to denote that the POVM element associated to \textit{accept} should be be separable, i.e.\ proportional to a separable state.} Then Harrow and Montanaro showed that $\QMA_{\SEP_{\text{YES}}}(2) = \QMA(2) = \QMA(k)$ for any $k = \poly(n)$ \cite{HM10}, i.e.\ two proofs are just as powerful as $k$ proofs and one can  restrict the verifier's action to $\SEP_{\text{YES}}$ without changing the expressive power of the class. 

We define $\QMA_{\LOCC}(k)$ in an analogous way, but now the verifier can only measure the $k$ proofs with a one-way LOCC measurement, with forward classical communication. Then we have,
\begin{corollary} \label{QMA2}
For $k = O(1)$,
\begin{equation} \label{QMA2LOCCequalQMA}
\QMA_{\LOCC}(k) = \QMA .
\end{equation}
In particular,
\begin{equation} \label{qma2exact}
\QMA_{\LOCC}(2)_{m, s, c} \subseteq \QMA_{ O(m^{2} \epsilon^{-2}), s + \epsilon, c} .
\end{equation}
\end{corollary}

A preliminary result in the direction of Corollary \ref{QMA2} appeared in \cite{Bra08}, where a similar result was shown for $\QMA_{\LO}(k)$, a variant of $\QMA(k)$ in which the verifier is restricted to implement only local measurements on the $k$ proofs and jointly post-process the outcomes classically.\footnote{$\QMA_{\LO}(k)$ is also called $\BellQMA(k)$ \cite{ABDFS08} since the verifier is basically restricted to perform a \textit{Bell test} on the proofs.}

It is an open question whether \eqref{qma2exact} remains true if we consider $\QMA(2)$ instead of $\QMA_{\LOCC}(2)$. If this turns out to be the case, then it would imply an optimal conversion of $\QMA(2)$ into $\QMA$ in what concerns the proof length (under a plausible complexity-theoretic assumption). For it follows from \cite{HM10} (based on the $\QMA(\poly \log (n)\sqrt{n})_{\log(n), 1, 1/3}$ protocol for 3-$\SAT$ of \cite{ABDFS08}) that unless there is a subexponential-time quantum algorithm for 3-$\SAT$, then there is a constant $\epsilon_0 > 0$ such that for every $\delta > 0$,
\begin{equation} 
\QMA(2)_{m, s, c} \nsubseteq \QMA_{ O(m^{2-\delta} \epsilon_{0}^{-2}), s + \epsilon_{0}, c}.
\end{equation}
Another open question is whether a similar result holds for $\QMA_{\LOCC}(2)$. This would be the case if one could construct a one-way LOCC version of the productness test of Harrow and Montanaro \cite{HM10}.

An interesting approach to the $\QMA(2)$ vs. $\QMA$ question concerns the existence of \textit{disentangler} superoperators \cite{ABDFS08}, defined as follows. 

\begin{definition}
A superoperator $\Lambda : S \rightarrow AB$ is a $(\log|S|, \epsilon, \delta)$-disentangler in the $\Vert * \Vert$ norm if 
\begin{itemize}
\item $\Lambda(\rho)$ is $\epsilon$-close to a separable state for every $\rho$, and
\item for every separable state $\sigma$, there is a $\rho$ such that $\Lambda(\rho)$ is $\delta$-close to $\sigma$.
\end{itemize}
\end{definition}
As noted in \cite{ABDFS08}, the existence of an efficiently implementable $(\poly(\log |A|, \log |B|), \epsilon, \delta)$-disentangler in trace norm (for sufficiently small $\epsilon$ and $\delta$) would imply $\QMA(2) = \QMA$. Watrous has conjectured that this is not the case and that for every $\epsilon, \delta < 1$, any $(\epsilon, \delta)$-disentangler (in trace-norm) requires $|S| = 2^{\Omega( \min(|A|, |B|))}$. For an exponentially large $|S|$, in turn, Matsumoto presented a construction of a quantum disentangler \cite{Mat05} and the quantum de Finetti theorem gives an alternative construction \cite{CKMR07}.
Corollary \ref{monogamy} can be rephrased as saying that there \textit{is} an efficient disentangler in LOCC norm. 
\begin{corollary}
Let $k = \lceil 8 \ln 2 \log |A| \epsilon^{-2} \rceil$ and $S := AB_1\cdots B_k$. Define the superoperator $\Lambda : S \rightarrow A  B$, with $|B| = |B_j|$ for all $j \leq k$, as
\begin{equation}
\Lambda(\rho_{AB_1\cdots B_k}) := \frac{1}{k} \sum_{i=1}^{k} \rho_{AB_i
}.
\end{equation}
Then $\Lambda$ is a $(O(\log | A| \log | B| \epsilon^{-2}), \epsilon, 0)$-disentangler in one-way LOCC norm.  
\end{corollary}

For more applications of this work to quantum Merlin-Arthur games, see~\cite{BCY11}.

\noindent {\bf Distinguishability under locality-restricted measurements.}
We will now give a formal definition of the norm $\norm{*}_\LOCC$ that was used in order to state the results. $\norm{*}_\LOCC$ can be seen as a restricted version of the trace norm $\norm{*}_1 $. Introducing the set $\ALL$ as consisting of all operators $M$ that satisfy $0 \leq M \leq \id$, the trace norm can be written as
\begin{align}\norm{X}_1= \max_{M \in \ALL} \tr((2M - \id) X),
\end{align} 
where $\id$ is the identity matrix.  The trace norm is of special importance in quantum information theory as it is directly related to the optimal probability for distinguishing two equiprobable states $\rho$ and $\sigma$ with a quantum measurement.\footnote{Two-outcome measurements suffice for such tasks, and these are described by a pair of positive semidefinite matrices summing to $\id$, which we write $\{M,\id - M\}$.  When the state is $\rho$, the probabilities of the outcomes are $\Pr(M) = \tr (M\rho)$ and $\Pr(\id - M) = \tr((\id - M)\rho) = 1 - \Pr(M)$. The optimal bias of distinguishing two states $\rho$ and $\sigma$ is then given by $\max_{0 \leq M \leq \id}\tr(M(\rho - \sigma)) = \frac{1}{2}\Vert \rho - \sigma \Vert_{1}$.}  
In analogy with this interpretation of the trace norm, and for operators $X$ on the tensor product space $AB$, we define the one-way LOCC norm as
\begin{equation}
\norm{X}_{\LOCC} := \max_{M \in \LOCC} \tr((2M - \id) X),
\end{equation}
where $\LOCC$ is the convex set of matrices $0 \leq M \leq \id$ such that there is a two-outcome measurement $\{M, \id - M \}$ that can be realized by one-way LOCC between Alice and Bob~\cite{MWW09}. It will become clear below why this defines a norm. The optimal bias in distinguishing $\rho$ and $\sigma$ by any one-way LOCC protocol is then $\frac{1}{2} \norm{\rho - \sigma }_{\LOCC}$. 
%We note that in many applications of the separability problem, e.g.\ assessing the usefulness of a quantum state for violating Bell's inequalities or for performing quantum teleportation, the LOCC norm is actually the more relevant quantity to consider. 

%In the proofs, we will also need versions of the LOCC norm, where the number of rounds of communication is restricted. Similarly to the definition of $\LOCC$ above, we therefore introduce $\rightLOCC(k)$ consisting of operators  $0 \leq M \leq \id$ such that there is a two-outcome measurement $\{M, \id - M \}$ that can be realized by $k$-round LOCC protocol with the first communication from Alice to Bob. A special role is played by the one-round implementable POVM elements $\rightLOCC\equiv \rightLOCC(1)$ for which an explicit characterisation exists (Lemma~\ref{lem:LOCCmeasurement}). Note that 
%\begin{align}
%\LOCC = \lim_{k \rightarrow \infty} \rightLOCC(k). 
%\end{align}
%Although not strictly needed in this work, one of the lemmas (Lemma~\ref{lowerboundnorm}) will also be valid for the set $\SEP$ which is defined as the convex hull of the set $\{M_j\ox N_j\}$, with $0\leq M_j,N_j \leq \id$. The inclusions of the sets then imply the following relations between the norms.
%\begin{align}
%\norm{X}_1 \geq \norm{X}_\SEP \geq \norm{X}_\LOCC \geq \cdots \geq \norm{X}_{\rightLOCC(k)} \geq \norm{X}_{\rightLOCC(k-1)} \geq \cdots \geq %\norm{X}_\rightLOCC.
%\end{align}
%Analogous statements hold true for $\leftLOCC(k)$, where the reverse arrow indicates that the first round of communication is from Bob to Alice.

The introduced norm fits into a general framework \cite{MWW09} for restricted norms where one considers 
$\Mclass$, a closed, convex subset of the operator interval 
\begin{align}
 \{M : 0 \leq M \leq \id\} = \{M : M \geq 0, \norm{M}_\infty \leq 1\}
 \end{align}
containing $\id$, having nonempty interior, and such that $M \in \Mclass$ implies $\id -M \in \Mclass$. Then
\begin{align} \label{eq:norm}
\norm{X}_\Mclass = \max_{M \in \Mclass}\tr\big(\!(2M - \id) X\big) = \max_{Y \in \mathsf{D}} \tr (YX),
\end{align}
where 
\[\mathsf{D} = \mathrm{conv}(\Mclass \cup -\Mclass) = \{2M- \id: M\in \Mclass\}= \{Y : \norm{Y}_\Mclass^* \leq 1\}
\] 
is the unit ball for the dual norm $\norm{\cdot}_\Mclass^*$.

\vspace{0.2 cm}
\noindent \textbf{Relating entanglement measures:} 
The core of the proof of the Theorem is composed of three steps, Lemmas~\ref{nonlockability}, \ref{almostmonogamy}, and~\ref{lowerboundnorm} below, each of which is a new result about entanglement measures. They ca be straightforwardly combined to prove the Theorem. 

A measure that plays a fundamental role in the proof is the regularised relative entropy of entanglement \cite{VP98, VW01}, defined as 
\begin{equation}
E_{R}^{\infty}(\rho_{A:B}) := \lim_{n \rightarrow \infty} \frac{E_{R}(\rho_{A:B}^{\otimes n})}{n},
\end{equation}
with the relative entropy of entanglement given by
\begin{equation}
E_R(\rho_{A:B}) := \min_{\sigma_{AB} \in {\cal S}_{A:B}} S(\rho_{AB} || \sigma_{AB}),
\end{equation}
where the quantum relative entropy is given by $S(\rho || \sigma) = \tr(\rho(\log \rho - \log \sigma))$ if $\supp \ \rho \subseteq \supp \  \sigma $ and  $S(\rho || \sigma)=\infty$ otherwise.
A distinctive property of the relative entropy of entanglement among entanglement measures is the fact that it is not lockable, i.e.\ by discarding a small part of the state, $E_R$ can only drop by an amount proportional to the number of qubits traced out. Indeed, as shown in \cite{HHHO05}, one has
\begin{equation} \label{nonlockable}
E_R(\rho_{A:BE}) \leq  E_R(\rho_{A:E}) +   2 H(B)_{\rho}.
\end{equation}
While the same is true for $E_R^{\infty}$, we use an optimal protocol for state redistribution \cite{DY07,YD07} to obtain the following improvement of Eq. \eqref{nonlockable}, which is valid only for the regularised quantity. 
\begin{lemma} \label{nonlockability}
Let $\rho_{ABE}$ be a quantum state. Then, 
\begin{equation}
 I(A;B|E)_{\rho} \geq E_R^\infty(\rho_{A:BE}) -  E_R^\infty(\rho_{A:E}).
\end{equation}
\end{lemma}

The second step in the proof is to obtain a lower bound on $E^{\infty}_{R}(\rho_{A:BE}) - E_R^\infty(\rho_{A:E})$, which by Lemma~\ref{nonlockability} also gives a lower bound on the conditional mutual information. If we could prove that $E_R^{\infty}$ satisfied the monogamy inequality given by \eqref{monogamyinequality}, we would be done, since Ref. \cite{Pia09} shows that 
\begin{equation}
E_{R}^{\infty}(\rho_{A:B}) \geq \frac{1}{2 \ln 2}  \Norm{ \rho_{A:B} - \CS_{A:B}}_{\leftrightLOCC}^{2}.
\end{equation}
However, the antisymmetric state is a counterexample \cite{CSW08}. Nonetheless, we will show that the regularised relative entropy of entanglement satisfies a weaker form of the monogamy inequality, and this will suffice to us. 

In order to state the new inequality, we have to recall an operational interpretation of $E_R^{\infty}$ in the context of quantum hypothesis testing \cite{BP10} (see Section \ref{sectionmonoghamyEr} for more details). Suppose Alice and Bob are given either $n$ copies of an entangled state $\rho_{A:B}$, or an arbitrary separable state in the cut $A^{n} \!:\! B^{n}$.  Suppose further they can only measure a POVM from the class $\Mclass$ in order to find out which of the two hypotheses is the case. Let $p_e(n)$ be the probability that Alice and Bob mistake a separable state for the $n$ copies of $\rho_{A:B}$, minimised over all possible measurements available to them.\footnote{We also require that the measurement only mistakes $\rho_{A:B}^{\otimes n}$ for a separable state with a small probability. See Section \ref{sectionmonoghamyEr} for details.} Then we define the optimal rate function $D_{\Mclass}(\rho_{A:B})$ as follows
\begin{equation}
D_{\Mclass}(\rho_{A:B}) = \lim_{n \rightarrow \infty} - \frac{\log p_e(n)}{n}.
\end{equation}
The main result of Ref. \cite{BP10} implies 
\begin{equation} \label{htfornetangelemnt}
D_{\ALL}(\rho_{A:B}) = E_{R}^{\infty}(\rho_{A:B}),
\end{equation}
i.e.\ the regularised relative entropy of entanglement is the optimal distinguishability rate when trying to 
distinguish many copies of an entangled state from (arbitrary) separable states, in the case where there is no restrictions 
on the measurements available. 

The new monogamy-like inequality we prove for $E_R^{\infty}$, using the connection of $E_R^{\infty}$ with hypothesis testing given by Eq. \eqref{htfornetangelemnt}, is the following:
\begin{lemma}\label{almostmonogamy}
For every state $\rho_{A:BE}$,
\begin{equation} 
E_{R}^{\infty}(\rho_{A:BE}) - E_{R}^{\infty}(\rho_{A:E})  \geq  D_{\leftLOCC}(\rho_{A:B}).
\end{equation}
\end{lemma}

The third and final step is to obtain a lower bound for $D_{\leftLOCC}$ in terms of the minimum $\leftLOCC$ distance to the set of separable states. We prove a slightly more general result. Then we have

\begin{lemma} \label{lowerboundnorm}
For any $\Mclass \in \{ \rightLOCC, \leftLOCC, \leftrightLOCC, \SEP \}$ and every state $\rho_{A:B}$
\begin{equation}
D_{\Mclass}(\rho_{A:B}) \geq \frac{1}{8 \ln 2} \Norm{\rho_{A:B} - \CS_{A:B}}_{\Mclass}^{2}.
\end{equation}
\end{lemma}

In \cite{Pia09} a similar result was proved: consider the following modification of the relative entropy of entanglement:
\begin{equation}
E_{R, \Mclass}(\rho) := \min_{\sigma \in {\cal S}}  \max_{M \in \Mclass}  S(M(\rho) || M(\sigma)),
\end{equation}
with $M(X) := \sum_k \tr(M_k X) \proj{k}$ for a POVM $M := \{ M_k \}_{k} \in \Mclass$. Then for any class of measurements $\Mclass \in \{\rightLOCC, \leftLOCC, \leftrightLOCC, \SEP \}$ it was shown that \cite{Pia09}
\begin{equation} \label{piani}
E_{R, \Mclass}(\rho_{A:B}) \geq \frac{1}{2 \ln 2}\Norm{\rho_{A:B} - \CS_{A:B}}_{\Mclass}^{2}.
\end{equation}
It is conceivable that 
\begin{equation} \label{eq?}
E_{R, \Mclass}^{\infty}(\rho_{A:B}) = D_{\Mclass}(\rho_{A:B}),
\end{equation}
and if this was the case we could obtain Lemma \ref{lowerboundnorm} from Eq. \eqref{piani}, since $E_{R, \Mclass}$ is superadditive. Indeed, for $\Mclass = \ALL$, Eq. \ref{eq?} holds true and this is precisely the content of Eq. \eqref{htfornetangelemnt}. For restricted classes of POVMs, however, it is not known whether Eq. \eqref{eq?} remains valid and, alas, the techniques of \cite{BP10} do not appear to be directly applicable. We note that one direction is easily seen to hold, namely\footnote{The proof is very similar to the usual argument, which shows $S(\rho || \sigma)$ to be an upper bound on the optimal distinguishability rate in quantum hypothesis testing (see e.g.\ \cite{HP91, ON00}), and thus we omit it.}
\begin{equation} 
E_{R, \Mclass}^{\infty}(\rho_{A:B}) \geq D_{\Mclass}(\rho_{A:B}).
\end{equation}
So we can obtain an asymptotic version of Eq. \eqref{piani} from Lemma \ref{lowerboundnorm}.

\vspace{0.2 cm}
\noindent \textbf{A new construction of entanglement witnesses:} In the proof of Lemma \ref{lowerboundnorm} we derive a relation that might be of independent interest. For any restricted set of POVMs $\Mclass$,
\begin{equation} \label{witness}
2\max_{M \in \Mclass} \left( \tr(M\rho) - \max_{\sigma \in {\cal S}} \tr(M \sigma) \right) = 
\Norm{ \rho_{AB} - \CS_{A:B}}_{\Mclass}. 
\end{equation}
The POVM element maximising the left hand side is a so-called \textit{entanglement witness} \cite{HHHH09, GT09}: its expectation value on the entangled state $\rho$ is bigger than its expectation value on any separable state. Therefore it can be used as a witness to attest the entanglement of $\rho_{AB}$. \eqref{witness} gives a new construction of entanglement witnesses that is attractive in the following respects: 

\begin{itemize}
\item  The witness is a POVM element and thus can be measured with a \textit{single} measurement setting.
\item Every $\Mclass$ that (after rescaling) contains an informationally complete POVM can be used to construct witnesses detecting every entangled state.
\item The maximum gap in expectation value for an entangled state versus separable states, optimised over all witnesses from a class $\Mclass$, is equal to the minimum distance of the entangled state to the set of separable states in the $\Mclass$ norm and thus can also be used to quantify the entanglement of the state (see \cite{Bra05, Bra07} for related results in this direction).
\end{itemize}

\section{Proof of Theorem}
\label{sec:mainproof}

%%%%%%%%%%%%%%%%%%%%%%%%
% Section II : Proof of Lemma 1
%%%%%%%%%%%%%%%%%%%%%%%%
\subsection{Proof of Lemma~\ref{nonlockability}}
%State redistribution and the relative entropy of entanglement: 
The main idea in the proof of Lemma~\ref{nonlockability} is to apply the nonlockability of $E_R$ given by \eqref{nonlockable} to a tensor power state $\rho_{A:BE}^{\ox n}$, but to use an optimal protocol for performing state redistribution protocol to trace out $B$ in a more efficient way.  
A general protocol for state redistribution is outlined in Figure~\ref{fig:redist}.  The following proposition follows immediately from the existence of a redistribution protocol achieving the minimum possible rates of communication and entanglement cost:
\begin{figure}
\begin{center}
\includegraphics[scale=.7]{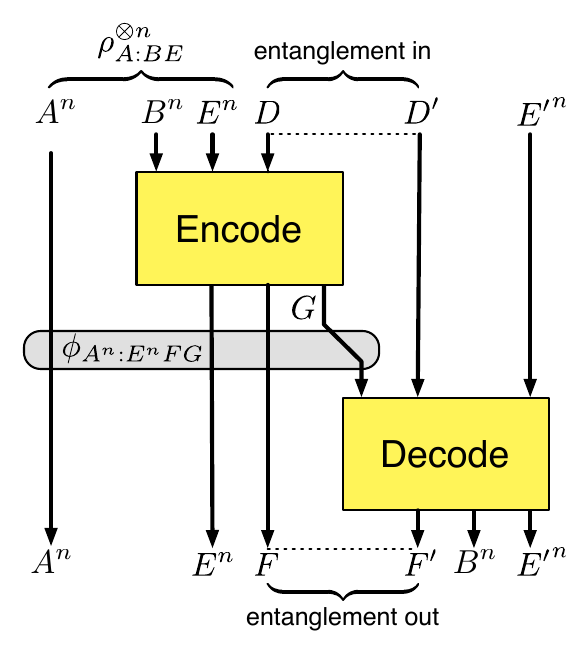}
\end{center}
\caption{General protocol for redistributing the $B$ parts of a pure state $\ket{\psi}_{ABEE'}^{\ox n}$ where the sender holds $E$ and the receiver holds $E'$. The protocol consumes maximal entanglement on $DD'$, transmits the quantum system $G$, and produces maximal entanglement on $FF'$.  The dimensions of $D$, $F$ and $G$ typically depend exponentially on $n$.  The protocol is asymptotically reversible.} 
\label{fig:redist}
\end{figure}

\begin{proposition}[\cite{DY07, YD07}]
There is a sequence of asymptotically reversible encoding maps $\Lambda_n: B^n E^n D \to E^n F G $ such that if 
$\phi_{A^n:E^n FG}= \Lambda_n (\rho_{A:BE}^{\ox n} \ox \tau_D)$ (see {\em Fig.~\ref{fig:redist}}), then
\begin{equation} \label{rate}
\frac{\log |G|}{n} \to  \smfrac{1}{2} I(A;B|E)_{\rho},
\end{equation}
and 
\begin{equation} \label{approx}
\Norm{\phi_{A^nE^nF} - \rho_{AE}^{\ox n} \ox \tau_{F}}_1 \to 0,
\end{equation}
where $\tau_D$ and $\tau_F$ are the maximally mixed states on systems $D$ and $F$, respectively.
\end{proposition}
\begin{proof}[{\bf Proof of Lemma~\ref{nonlockability}}]
\begin{eqnarray}
E_R^\infty(\rho_{A:BE}) & = & \lim_{n \rightarrow \infty} \frac{1}{n} E_R(\phi_{A^n : E^nFG}) \\ 
& \leq & \left(\lim_{n \rightarrow \infty} \frac{1}{n}  E_R(\phi_{A^n : E^n F})\right) +  I(A;B|E) \\ 
&= &  \left(\lim_{n \rightarrow \infty} \frac{1}{n}  E_R(\rho^{\ox n}_{A : E} \ox \tau_{F}) \right)+ I(A;B|E)\\ 
&=& E_R^\infty(\rho_{A:E})+ I(A;B|E).  
\end{eqnarray}
The first line holds by asymptotic continuity of $E_R$ \cite{DH99, SH06} and the asymptotic reversibility of the encoding operation.  The second follows from  (\ref{nonlockable}) and (\ref{rate}).  The third line holds again by continuity of $E_R$, applied to the asymptotic equality (\ref{approx}) of the respective arguments of $E_R$.
\end{proof}

The same argument applied to the mutual information instead of the relative entropy of entanglement can be used to prove the converse of state redistribution (see~\cite{Berta10, Berta11} for a similar argument in the context of one-shot quantum state splitting).

\subsection{Proof of Lemma~\ref{almostmonogamy}} \label{sectionmonoghamyEr}
%Entanglement Hypothesis testing: 
Before proceeding to the proof, we give a precise definition of the rate functions $D_{\Mclass}$. 
\begin{definition} 
Let $\rho_{AB}$ be a state and $\Mclass=\{\rightLOCC, \leftLOCC,  \LOCC, \SEP, \ALL\}$ a class of POVMs. Then we define $D_{\Mclass}(\rho_{A:B})$ as the unique non-negative number such that
\begin{itemize}
\item (\textit{Direct part}): For every  $\epsilon > 0$ there exists a sequence of POVM elements $\{M_{A^nB^n}\}_{n \in \mathbb{N}}$, with $M_{A^nB^n} \in \Mclass$, such that
\begin{equation}
        \lim_{n \rightarrow \infty} \tr(M_{A^nB^n} \rho_{AB}^{\otimes n}) = 1
\end{equation}
and for every $n \in \mathbb{N}$ and $\omega_{A^nB^n} \in {\cal S}_{A^{n}:B^{n}}$,
\begin{equation}
        - \frac{\log \tr(M_{A^nB^n} \omega_{A^nB^n})}{n} \geq 
        D_{\Mclass}(\rho_{A:B}) - \epsilon.
\end{equation}
\item (\textit{Converse}): If a real number $\epsilon > 0$ and a sequence of POVM elements 
$\{ M_{A^nB^n} \}_{n \in \mathbb{N}}$ are such that for every $n \in \mathbb{N}$ and every $\omega_{A^nB^n} \in {\cal S}_{A^{n}:B^{n}}$,
\begin{equation}
         - \frac{\log( \tr(M_{A^nB^n} \omega_{A^nB^n}))}{n}  \geq 
         D_{\Mclass}(\rho_{A:B}) + \epsilon,
\end{equation}
then
\begin{equation}
        \lim_{n \rightarrow \infty} \tr(M_{A^nB^n}\rho_{AB}^{\otimes n}) < 1. 
\end{equation}
\end{itemize}
\end{definition}
\noindent
Recall~\eqref{htfornetangelemnt}:
\begin{proposition}[\cite{BP10}] \label{stein}
For every state $\rho_{A:B}$,
\begin{equation}
D_{\ALL}(\rho_{A:B}) = E_{R}^{\infty}(\rho_{A:B}).
\end{equation}
\end{proposition}

 The main application of Proposition~\ref{stein} so far has been to prove that the asymptotic manipulation of entanglement under operations that, approximately, do not generate entanglement is reversible: both the distillation and a cost functions in this paradigm are equal to the regularised relative entropy of entanglement \cite{BP08, BP10b, BD09}. In fact, one can show that Proposition~\ref{stein} is \textit{equivalent} to such reversibility under non-entangling operations, in the sense that the latter can be used to prove the former \cite{Bra10}. 

The proof of Lemma~\ref{almostmonogamy} will utilise the special structure of LOCC measurements with one-way communication:
\begin{lemma}
Any operator in $\leftLOCC$ necessarily has the form   
\begin{align} \sum_k M_{A, k} \ox N_{B, k},\end{align}
where $0 \leq M_{A, k} \leq \mathbb{I}_A$,  $0\leq N_{B, k}$ and $\sum_k N_{B, k} \leq \mathbb{I}_B$.
\label{lem:LOCCmeasurement}
\end{lemma}
\begin{proof}
Any LOCC measurement consisting of communication only from Bob to Alice has the general form $\{M_{A, a|b} \ox N_{B, b}\}$, where  $0\leq M_{A, a|b}$, $0\leq N_{B, b}$, $\sum_a M_{A, a|b} \leq \mathbb{I}_A$ and $\sum_b N_{B, b} \leq \mathbb{I}_B$.  The convex hull of such POVM elements consists of all operators of the form given in the lemma. 
\end{proof}
\begin{proof}[{\bf Proof of Lemma~\ref{almostmonogamy}}.] 
The main idea in this proof is to start with measurements for $\rho_{A:E}$ and $\rho_{A:B}$ that respectively achieve $D_{\ALL}(\rho_{A:E})$ and $D_{\leftLOCC}(\rho_{A:B})$, and to construct a measurement for $\rho_{A:BE}$ distinguishing it from separable states at a rate given by the sum of the individual rates. The lemma then follows from Proposition~\ref{stein}. Here, it is crucial that we only use one-way LOCC measurements to discriminate $\rho_{A:B}$ from separable states, since only in this case are we able to make sure that the composed measurement for $\rho_{A:BE}$ acts as it should. 

Let $\delta > 0$ and fix two $\delta$-optimal sequences of POVM elements: an unrestricted sequence $Q_{A^nE^n}$ of POVM elements for distinguishing $\rho_{A:E}$ from separable states, and a sequence $S_{A^nB^n}$ of one-way LOCC (with classical communication from Bob to Alice) elements for distinguishing $\rho_{A:B}$ from separable states.  
This means that 
\begin{equation} 
\tr \left( Q_{A^nE^n} \rho^{\ox n}_{A:E} \right) \to 1 \text{ and } \tr \left( S_{A^nB^n} \rho^{\ox n}_{A:B} \right) \to 1,
 \label{projectwellSonrho}
\end{equation}
while 
\begin{equation}
-\frac{\log \tr \left( Q_{A^nE^n} \om^n_{A:E} \right)}{n} \geq D_{\ALL}(\rho_{A:E}) - \delta
\label{decay1ALL}
\end{equation}
for every state $\om_{A^n:E^n}$ separable across $A^n\!\!:\!E^n$ and 
\begin{equation}
-\frac{\log \tr \left( S_{A^nB^n} \om_{A^n:B^n} \right)}{n} \geq D_{\leftLOCC}(\rho_{A:B}) - \delta
\label{decay1LOCC}
\end{equation}
for every state $\om_{A^n:B^n}$ separable across $A^n\!\!:\!B^n$.
We will use these to construct a sequence of POVM elements $T_{A^nB^nE^n}$ that distinguish $\rho_{A:BE}^{\otimes n}$ from separable states $\om_{A^n:B^nE^n}$, satisfying 
\begin{equation} \label{final_rho_condition}
\lim_{n \rightarrow \infty} \tr(T_{A^nB^nE^n} \rho_{A:BE}^{\otimes n}) = 1
\end{equation}
and 
\begin{equation}
\lim_{n \rightarrow \infty} - \frac{\log \tr(T_{A^nB^nE^n} \omega_{A^n:B^nE^n})}{n} \geq D_{\ALL}(\rho_{A:E}) + D_{\leftLOCC}(\rho_{A:B}).
\label{final_om_condition}
\end{equation}
Because the left-hand-side is a lower bound to the maximal achievable error rate $D_{\ALL}(\rho_{A:BE})$, Proposition~\ref{stein} then gives the desired inequality 
\begin{equation}
E_{R}^{\infty}(\rho_{A:BE}) \geq E_{R}^{\infty}(\rho_{A:E}) + D_{\leftLOCC}(\rho_{A:B}).
\nonumber
\end{equation}

%%%%%%%%%%%%%%%%%%%
% Construction of T^{ABE}
%%%%%%%%%%%%%%%%%%%
To construct $T_{A^nB^nE^n}$, first observe that by Lemma~\ref{lem:LOCCmeasurement}, we can write the $\leftLOCC$ operator $S_{A^nB^n}$ in the form
\begin{equation}
S_{A^nB^n}=\sum_{k} M_{A^n, k} \otimes N_{B^n, k},
\end{equation}
with $0\leq M_{A^n, k} \leq \id$, $0\leq N_{B^n, k}$ and $\sum_k N_{B^n, k} \leq \id$.  
It will be useful to consider a Naimark extension of the measurement on $B$, consisting of orthogonal projections $P_{B^n\wh{B}, k}$ such that 
$\bra{0_{\wh{B}}}P_{B^n\wh{B}, k}\ket{0_{\wh{B}}} = N_{B^n\wh{B}, k}$.
Define 
\begin{equation} \label{formR}
R_{A^{n}B^{n}\wh{B}} := \sum_k M_{A^n, k} \otimes P_{B^n\wh{B}, k}, 
\end{equation}
which by construction is such that $0 \leq R_{A^{n}B^{n}\wh{B}} \leq \id$.
Define 
\begin{align}
T_{A^nB^nE^n} & := \BRa{0_{\wh{B}}}  \sqrt{R_{A^nB^n\wh{B}}}  \,\, Q_{A^nE^n} \sqrt{R_{A^nB^n\wh{B}}} \KEt{0_{\wh{B}}} \label{eqnQ1} \\
& = \sum_k \left(\sqrt{M_{A^n, k}} \,\,Q_{A^nE^n} \sqrt{M_{A^n, k}}\right) \ox N_{B^n, k}. \label{eqnQ2}
\end{align}
Note that we omit identity operators from our notation here and in the remainder of this proof.  There should be no room for confusion, as the global Hilbert space will always be clear, and we use superscripts to indicate on which "local systems" each operator acts.  For instance, 
$Q_{A^n E^n} \equiv Q_{A^n E^n} \ox \id_{B^n \wh{B}}$ in Eq.\ (\ref{eqnQ1}), while $M_{A^n,k} \equiv M_{A^n,k} \ox \id_{E^n}$ in Eq.\ (\ref{eqnQ2}).

Observe that $0 \leq T_{A^nB^nE^n} \leq \mathbb{I}$ because the same holds for $Q_{A^nE^n}$ and $R_{A^nB^n\wh{B}}$. 
%%%%%
% Proof that T detects rho
%%%%%
We will now verify the conditions (\ref{final_rho_condition}) and (\ref{final_om_condition}) required to complete the proof.  
For proving \eqref{final_rho_condition}, we first write
\begin{equation} \label{auxmainproof1}
\tr(T_{A^nB^nE^n} \rho_{A:BE}^{\otimes n}) = \tr \left(Q_{A^nE^n}   \sqrt{R_{A^nB^n\wh{B}}} \left(\rho_{A:BE}^{\otimes n} \otimes \proj{0}_{\wh{B}} \right)\!\sqrt{R_{A^nB^n\wh{B}}} \, \right).
\end{equation}
From \eqref{projectwellSonrho}
\begin{equation}
\lim_{n \rightarrow \infty} \tr \Big(R_{A^nB^n\wh{B}} \left(\rho_{A:BE}^{\otimes n} \otimes \proj{0}_{\wh{B}} \right) \Big)= \lim_{n \rightarrow \infty} \tr(S_{A^nB^n} \rho_{A:B}^{\otimes n}) = 1.
\end{equation}
Hence by the gentle measurement lemma \cite{Win99,ON02},
\begin{equation}
\lim_{n \rightarrow \infty} \NORM{\rho_{A:BE}^{\otimes n} \otimes \proj{0}_{\wh{B}}  - \sqrt{R_{A^nB^n\wh{B}}} \left(\rho_{A:BE}^{\otimes n} \otimes\proj{0}_{\wh{B}} \right) \!\sqrt{R_{A^nB^n\wh{B}}} \,}_1 = 0
\end{equation}
and so \eqref{auxmainproof1} and \eqref{projectwellSonrho} give
\begin{align}
\lim_{n \rightarrow \infty} & \tr(T_{A^nB^nE^n} \rho_{A:BE}^{\otimes n}) 
\geq \lim_{n \rightarrow \infty} \tr \Big( Q_{A^nE^n}  \left(\rho_{A:BE}^{\otimes n} \otimes \proj{0}_{\wh{B}}  \right) \Big) \nonumber \\ 
 & \quad - \lim_{n \rightarrow \infty} \NORM{\rho_{A:BE}^{\otimes n} \otimes\proj{0}_{\wh{B}}  -  \sqrt{R_{A^nB^n\wh{B}}}  \left(\rho_{A:BE}^{\otimes n} \otimes\proj{0}_{\wh{B}} \right)  \!\sqrt{R_{A^nB^n\wh{B}}} \,}_1 \nonumber \\ 
&= 1.
\end{align}
%%%%%%%%%%
% Bound for T to detecting separable states
%%%%%%%%%%
Let us now prove \eqref{final_om_condition}. We write
\begin{equation} \label{productmain}
\tr(T_{A^nB^nE^n} \omega_{A^n:B^nE^n}) = \tr \left(Q_{A^nE^n} \tilde{\omega}_{A^n:E^n} \right) \tr \Big(R_{A^nB^n\wh B} \left(\omega_{A^n:B^nE^n} \otimes \proj{0}_{\wh{B}}  \right)\Big).
\end{equation}
with
\begin{equation}
\tilde{\omega}_{A^n:E^n} :=\frac{ \tr_{B^{n}\wh{B}}\Big(\sqrt{R_{A^nB^n\wh B}}  \left(\omega^n_{A:BE} \otimes \proj{0}_{\wh{B}}  \right)\sqrt{R_{A^nB^n\wh B}}\,\, \Big) }{ \tr \Big(R_{A^nB^n\wh B}  \left(\omega_{A^n:B^nE^n} \otimes\proj{0}_{\wh{B}}  \right)\Big)}.
\end{equation}
We claim that $\tilde{\omega}_{A^n:E^n}$ is separable in the cut $A^{n}\!:\!E^{n}$ whenever $\omega_{A^n:B^nE^n}$ is separable in 
the cut $A^{n}\!:\!B^{n}E^{n}$. Then \eqref{productmain} will imply \eqref{final_om_condition}because
\begin{eqnarray}
\lim_{n \rightarrow \infty} - \frac{\log \tr(T_{A^nB^nE^n} \omega_{A^n:B^nE^n})}{n} &=& \lim_{n \rightarrow \infty} - \frac{\log \tr(Q_{A^nE^n} \tilde{\omega}_{A^n:E^n})}{n} + \lim_{n \rightarrow \infty} - \frac{\log \tr(S_{A^nB^n} \omega_{A^n:B^n})}{n}  \nonumber \\ &\geq & E_{R}^{\infty}(\rho_{A:E}) + D_{\leftLOCC}(\rho_{A:B}) - 2 \delta,
\end{eqnarray}
where we used Proposition~\ref{stein} and \eqref{decay1ALL}) and  \eqref{decay1LOCC}.

To show that $\tilde{\omega}_{A^n:E^n}$ is separable, we now prove that  the completely-positive trace-non-increasing map $\Lambda\colon A^nB^nE^n \to A^nE^n$  defined as
\begin{equation}
\Lambda(X_{A^nB^nE^n}) := \tr_{B^{n}\wh{B}}\left(\sqrt{R_{A^nB^n\wh B}} \left(X_{A^nB^nE^n} \otimes\proj{0}_{\wh{B}} \right)\sqrt{R_{A^nB^n\wh B}}\,\, \right)
\end{equation}
is a separable superoperator. From the form of $R_{A^nB^n\wh{B}}$ given in \eqref{formR} we can write 
\begin{equation} \label{separablemap}
\Lambda(X_{A^nB^nE^n})
 :=  \sum_{k, k'}\tr_{B^{n}\wh{B}}\left[\left(\sqrt{M_{A^n,k}}\otimes P_{B^n\wh B,k}\right)  \left(X_{A^nB^nE^n} \otimes\proj{0}_{\wh{B}}  \right)\left(\sqrt{M_{A^n, k'}}\otimes P_{B^n\wh B, k'}\right)\right].\nn
\end{equation}
Then, taking the partial trace over $B^{n}\wh{B}$ we find that the cross terms in \eqref{separablemap} are zero since the projectors $P_{B^n\wh B, k}$ are mutually orthogonal and therefore
\begin{equation}
\Lambda (X_{A^nB^nE^n})
  := \sum_{k}\tr_{B^{n}\wh{B}}\left[\left(\sqrt{M_{A^n,k}}\otimes P_{B^n\wh B, k}\right)   \left(X_{A^nB^nE^n} \otimes\proj{0}_{\wh{B}}  \right)\left(\sqrt{M_{A^n,k}}\otimes P_{B^n\wh B, k}\right)\right], \nn
\end{equation}
which is manifestly a separable superoperator. This implies that $\tilde{\omega}_{A^n:E^n}=\Lambda(\omega_{A^n:B^nE^n})$ is separable, since $\omega_{A^n:B^nE^n}$ is separable.
\end{proof}

\subsection{Proof of Lemma~\ref{lowerboundnorm}}

The final step is now to prove Lemma \ref{lowerboundnorm}.  For that we construct a particular protocol for distinguishing many copies of an entangled state from arbitrary separable states, with an exponentially small probability of mistaking a separable state for many copies of a particular entangled state. In Corollary II.2 of \cite{BP10}, a similar result was shown using the exponential de Finetti theorem \cite{Ren07}. Here we give a simpler proof with a stronger bound on the distinguishability rate.

A useful result we will employ is the following simple extension of \cite{Jai05}, which in turn is a consequence of the minimax theorem, which we quote below the proof.
\begin{lemma} \label{jain}
Let $\Mclass$ be a closed convex set of POVM elements and ${\cal C}$ a closed convex set of states. Then given $\rho \notin {\cal C}$,
\begin{equation}
2\max_{M \in \Mclass} \left( \tr(M\rho) - \max_{\sigma \in {\cal C}} \tr(M \sigma) \right) = \norm{ \rho - \mathcal{C}}_{\Mclass}  
\end{equation}
\end{lemma}
\begin{proof}
Following Jain \cite{Jai05}, let the function $u : (\Mclass, {\cal C}) \rightarrow \mathbb{R}$ be defined as $u(M, \sigma) = \tr(M\rho) - \tr(M \sigma)$. The sets $\Mclass$ and ${\cal  C}$ are compact convex subsets of the (real) vector space of Hermitian matrices. Since the function $u$ is furthermore affine in both arguments, the conditions of the minimax theorem are satisfied (see Proposition \ref{minimax} just below) and we find
\begin{eqnarray}
2\max_{M \in \Mclass} \left( \tr(M\rho) - \max_{\sigma \in {\cal C}} \tr(M \sigma) \right) &=& 2\max_{M \in \Mclass} \min_{\sigma \in {\cal C}} \left( \tr(M\rho) - \tr(M \sigma) \right)  \\ &=&  2\min_{\sigma \in {\cal C}} \max_{M \in \Mclass} \left( \tr(M \rho) - \tr(M \sigma) \right)  \\ &=&  \min_{\sigma \in {\cal C}} \Vert \rho - \sigma \Vert_{\Mclass}. 
\end{eqnarray}
\end{proof}

\begin{proposition} \label{minimax}
(Minimax theorem) Let $\CC_1, \CC_2$ be non-empty, convex and compact subsets of $\mathbb{R}^{n}$ and  $\mathbb{R}^{m}$, respectively. Let $u : \CC_1 \times \CC_2 \rightarrow \mathbb{R}$ be convex in the first argument and concave in the second and continuous in both. Then,
\begin{equation}
\min_{x_1 \in \CC_1} \max_{x_2 \in \CC_2}  u(x_1, x_2)= \max_{x_2 \in \CC_2}  \min_{x_1 \in \CC_1}  u(x_1, x_2). 
\end{equation}
\end{proposition}

\vspace{0.3 cm}
In the proof of Lemma~~\ref{lowerboundnorm} we will also make use of Azuma's inequality, a large deviation bound for correlated random variables. We say a sequence of random variables $X_0, X_1, X_2, \dotsc$ is a supermartingale if for all $k\in \mathbb{N}$
\begin{equation}
\mathbb{E}\left( X_{k+1} | X_0, \dotsc, X_k \right) \leq X_k. 
\end{equation}
Then we have 
\begin{proposition} \label{azuma}
(Azuma's inequality) Let $\{ X_k \}_{k}$ be a supermartingale with $|X_{k+1} - X_k| \leq c_k$. Then for all positive integers $n$ 
and all positive reals $t$,
\begin{equation}
\Pr \left( X_n - X_0 \geq t  \right) \leq \exp \left(  - \frac{t^{2}}{2\left( \sum_{k=1}^{n} c_k^{2} \right)}  \right). 
\end{equation}
\end{proposition}

\begin{proof}[{\bf Proof of Lemma~\ref{lowerboundnorm}}] 
By Lemma \ref{jain}, there is $M_{AB} \in \Mclass$ such that 
\begin{equation} 
b = a + \frac 12 \Norm{\rho_{AB} - \CS_{A:B}}_{\Mclass}, \text{\,\, where \,\,} 
a:= \max_{\sigma_{AB} \in {\cal S}_{A:B}} \tr(M_{AB}\sigma_{AB}), \hspace{0.2 cm} b := \tr(M_{AB}\rho_{AB}),
\label{jainconsequence}
\end{equation}
and without loss of generality 
$\Norm{\rho_{AB} - \CS_{A:B}}_{\Mclass}>0$, or $b>a$.

The element $M_{A^nB^n} \in \Mclass$ is constructed as follows. Fix $\delta > 0$. One measures the POVM $\{ M_{AB}, \id - M_{AB} \}$ 
in each of the $n$ copies, obtaining $n$ binary random variables $\{ Z_i \}_{i=1}^{n}$, where $Z_i=1$ is associated with effect $M_{AB}$ and $Z_i=0$ with effect $\id-M_{AB}$. Then if 
\begin{equation}
\frac{1}{n} \sum_{i=1}^{n} Z_i \geq a + (1 - \delta)(b - a), 
\end{equation}
we accept (this corresponds to $M_{A^nB^n}$). Otherwise we reject. 

First, by the law of large numbers it is clear that 
\begin{equation}
\lim_{n \rightarrow \infty} \tr(M_{A^nB^n} \rho_{AB}^{\otimes n}) = 1. 
\end{equation}
We now show that
\begin{equation}
\lim_{n \rightarrow \infty} - \frac{\log \tr(M_{A^nB^n} \omega_{A^nB^n})}{n} \geq \frac{ (1 - \delta)^{2}}{8 \ln 2} \Norm{\rho_{AB} - \CS_{A:B}}_{\Mclass}^{2}. 
\end{equation}
The key observation is that, in case we have a separable state $\omega_{A^nB^n}$, for all $i \in \{ 1, \dotsc, n\}$,
\begin{equation}
\Pr(Z_i = 1 | Z_1, \dotsc, Z_{i-1}) \leq a.  
\end{equation}
This is so because by assumption $M_{AB} \in \SEP$ implies that both POVM elements in the measurement $\{ M_{AB}, \id - M_{AB} \}$ are separable operators. Therefore, the state to which the $k$-th measurement is applied is always separable, even conditioning on previous measurement outcomes. Then by Proposition~\ref{azuma},
\begin{equation}
\Pr \left( \frac{1}{n} \sum_{i=1}^{n} Z_i \geq a + (1 - \delta)(b - a) \right) \leq 2^{ - \frac{ (1 - \delta)^{2}}{2 \ln 2} n (b-a)^{2} }. 
\end{equation}
In more detail, define the random variables
\begin{equation}
X_k := \sum_{i=1}^{k}\left( Z_i - a \right). 
\end{equation}
Note that
\begin{eqnarray}
\mathbb{E} \left( X_k | X_1, \dotsc, X_{k-1}  \right) &=& \mathbb{E} \left( Z_{k} | Z_1, \dotsc, Z_{k-1}  \right) - a + X_{k-1}  \\ 
&=& \Pr \left( Z_{k} = 1 | Z_1, \dotsc, Z_{k-1}  \right) - a + X_{k-1}  \\ 
&\leq& X_{k-1}, 
\end{eqnarray}
i.e.\ $\{ X_k \}_{k}$ form a supermartingale. Moreover, for all $k$, $|X_k - X_{k-1}| = |Z_{k} - a| \leq 1$. Hence by Proposition \ref{azuma}, setting $X_0=0$, we have
\begin{eqnarray}
\Pr \left( \frac{1}{n} \sum_{i=1}^{n} Z_i \geq a + (1 - \delta)(b - a) \right) &=& \Pr \left( X_{n} \geq n (1 - \delta)(b - a) \right) \\ &\leq&  \exp \left( - \frac{(1 - \delta)^{2}}{2} n(b-a)^{2}\right) \\ &=& 2^{ - \frac{ (1 - \delta)^{2}}{2 \ln 2} n (b-a)^{2}}. 
\end{eqnarray}
From \eqref{jainconsequence},
\begin{equation}
- \frac{1}{n} \log \left( \Pr \left( \frac{1}{n} \sum_{i=1}^{n} Z_i \geq a + (1 - \delta)(b - a) \right) \right) \geq \frac{(1 - \delta)^{2}}{ 8 \ln 2}\Norm{\rho_{AB} - \CS_{A:B}}_{\Mclass}^{2}, 
\end{equation}
and thus
\begin{equation}
D_{\Mclass}(\rho_{A:B}) \geq \frac{(1 - \delta)^{2} }{8 \ln 2}\Norm{\rho_{AB} - \CS_{A:B}}_{\Mclass}^{2}. 
\end{equation}
The result then follows from the fact that $\delta > 0$ is arbitrary.
\end{proof}

\subsection{Proof of Theorem} \label{onetotwo}
%From one-way LOCC norm to two-way LOCC norm: 

\begin{proof}
The Theorem follows directly from Lemmas 1, 2, and 3. Indeed
\begin{equation}
I(A; B|E)_{\rho} \geq E_{R}^{\infty}(\rho_{A:BE}) - E_R^{\infty}(\rho_{A:E}) \geq D_{\rightLOCC}(\rho_{A:B}) \geq \frac{1}{2 \ln 2} \Vert \rho_{A:B} - {\cal S}_{A:B} \Vert_{\rightLOCC}^{2},
\end{equation}
where the first, second, and third inequalities follow from Lemmas 1, 2, and 3, respectively.
\end{proof}

\section{ Proofs of Corollaries} \label{twocorollaries}
We start with the proof of Corollary \ref{monogamy}, which is straightforward.

\begin{proof}[{\bf Proof of Corollary \ref{monogamy}}]

Squashed entanglement satisfies the monogamy inequality \cite{KW04}
\begin{equation}
E_{\rm sq}(\rho_{A:BB'}) \geq E_{\rm sq}(\rho_{A:B}) + E_{\rm sq}(\rho_{A:B'}).
\end{equation}
Repeatedly applying it to $\rho_{A:B_1 \cdots B_k}$ and noting that $E_{\rm sq}(\rho_{A:B_1 \cdots B_k}) \leq \log |A|$~\cite{CW04}, we get
\begin{equation}
\log|A| \geq E_{\rm sq}(\rho_{A:B_1...B_k}) \geq kE_{\rm sq}(\rho_{A:B}). 
\end{equation}
where we used that $\rho_{AB_j}=\rho_{AB}$ for all $j$. The result then follows from the lower bound for squashed entanglement given in Corollary \ref{corsquahsedlower}.
\end{proof}

Here, we give an outline of the proof of Corollary~\ref{complexity}, concerning the quasipolynomial-time algorithm for deciding separability.  A more careful treatment can be found in \cite{BCY11}. 
\begin{proof}[{\bf Proof of Corollary \ref{complexity}}]

Given a density matrix $\rho_{AB}$, we use semidefinite programming \cite{VB96} in order to search for a $\left\lceil \frac{8 \ln2 \log|A|}{\epsilon^{2}}\right\rceil$-extension of $\rho_{A:B}$. If $\rho_{AB} \in \CS_{A:B}$, then there is such an extension (since separable states have a $k$-extension for every $k$). If $\Norm{\rho_{AB} - \CS_{A:B}}_\LOCC \geq \ep$, Corollary \ref{monogamy} implies that there is no such extension.

The computational cost of solving a feasibility semidefinite problem with $n$ variables and a matrix of dimension $m$ is $O(n^{2}m^{2})$. A $k$-extension of $\rho_{AB}$ has dimension $|A||B|^{k}$ and less than $|A|^{2}|B|^{2k}$ variables. Hence the time complexity of searching for a $\left\lceil \frac{8 \ln2 \log|A|}{\epsilon^{2}}\right\rceil$-extension is of order $\exp \left( O(\epsilon^{-2} \log|A| \log|B|)\right)$.
\end{proof}

%\section{Proof of Corollary \ref{QMA2}}
%Quantum Merlin-Arthur games: 

\begin{proof}[{\bf Proof of Corollary \ref{QMA2}}]

We start by proving \eqref{qma2exact}. Consider a protocol in $\QMA_{\LOCC}(2)_{m, s, c}$ given by the one-way LOCC measurement $\{  M_{AB}, \id - M_{AB}\}$. We construct a protocol in $\QMA_{O(m^{2}\epsilon^{-2}), s + \epsilon, c}$ that can simulate it: The verifier asks for a proof of the form $\ket{\psi}_{AB_1\cdots B_k}$ where $|A| = |B_i| = 2^{m}$ for all $i$ (each register consists of $m$ qubits) and $k = \lceil 8 \ln2 \epsilon^{-2} m \rceil $. He then symmetrises the $B$ systems obtaining the state $\rho_{AB_1\cdots B_k}$ and measures $\{ M_{AB}, \id - M_{AB}\}$ in the subsystems $AB\equiv AB_1$. 

Let us analyse the completeness and soundness of the protocol. For completeness, the prover can send $\ket{\psi}_{A} \otimes \ket{\phi}_B^{\otimes k}$, for states $\ket{\psi}, \ket{\phi}$ such that
\begin{equation}
\tr \proj{\psi}_A \otimes \proj{\phi}_B M_{AB} \geq c.
\end{equation}
Thus the completeness parameter of the $\QMA$ protocol is at least $c$.

For soundness, we note that by Corollary \ref{monogamy}, 
\begin{equation} \label{eqclose}
\Norm{\rho_{AB} - \CS_{A:B}}_{\LOCC} \leq \epsilon.
\end{equation}
Thus, as $\{ M_{AB}, \id - M_{AB}\} \in \leftrightLOCC$ the soundness parameter for the $\QMA$ protocol can only be $\epsilon$ away from $s$. Indeed, for every $\rho_{AB_1\cdots B_k}$ symmetric in $B$,
\begin{equation}
\tr(M\rho_{AB}) \leq \max_{\sigma_{AB} \in {\cal S}_{A:B}} \tr(M_{AB} \sigma_{AB}) +  \Norm{\rho_{AB} - \CS_{A:B}}_{\LOCC} \leq s + \epsilon.
\end{equation}
\eqref{QMA2LOCCequalQMA} now follows easily from the protocol above. Given a protocol in $\QMA_{'\LOCC}(l)$ with each proof of size $m$ qubits, we can simulate it in $\QMA_{\LOCC}(l-1)$ as follows: The verifier asks for $l-1$ proofs, the first proof consisting of registers $AB_1\cdots B_k$, each of size  $m$ qubits and $k =\lceil  8 \ln2 m \epsilon^{-2}\rceil $, and all the $l-2$ other proofs in systems $C_j$ ($3\leq j\leq l$) of size $m$ qubits. Then he symmetrises the B systems and traces out all of them except the first. Finally he applies the original measurement from the $\QMA_{\LOCC}(l)$ protocol to the resulting state.

The completeness of the protocol is unaffected by the simulation. For the soundness let $\rho_{AB_1\cdots B_k} \otimes \rho_{C_3} \otimes \cdots \otimes \rho_{C_l}$ be an arbitrary state sent by the prover (after the symmetrisation step has been applied to $B_1, \dotsc, B_m$). Let $\{ M_{C^l}, \id - M_{C^l}  \} \in \LOCC$ be the $l$-partite LOCC verification measurement from the $\QMA_{\LOCC}(l)$ protocol (where for notational simplicity we define $C_1=A$, $C_2=B_1$ and $C^l=C_1\cdots C_l$). Then
\begin{eqnarray}
\tr \left(M_{C^l} (\rho_{C_1C_2}\otimes \rho_{C_3} \otimes \cdots \otimes \rho_{C_l}) \right) 
&\leq& \max_{\sigma_{C^l} \in {\cal S}_{C_1:\cdots:C_l}} \tr(M_{C^l} \sigma_{C^l})\\
&& \qquad + \min_{\sigma_{C^l} \in {\cal S}_{C_1:\cdots:C_l}} \Vert \rho_{C_1C_2} \otimes \rho_{C_3} \otimes \cdots \otimes \rho_{C_l} - \sigma_{C^l} \Vert_{\LOCC} \nn \\ 
&=& \max_{\sigma_{C^l} \in {\cal S}_{C_1:\cdots:C_l}} \tr(M_{C^l} \sigma_{C^l}) \\
&&  + \min_{\sigma_{C_1C_2} \in {\cal S}_{C_1:C_2}} \Vert \rho_{C_1C_2} - \sigma_{C_1C_2} \Vert_{\LOCC} \nn \\ &\leq& s + \epsilon,
\end{eqnarray}
The equality in the second line follows since we can assume that the states $\rho_{C_3}, \cdots , \rho_{C_l}$ belong to Alice and adding local states does not change the minimum one-way LOCC-distance to separable states. 

Since for going from $\QMA_{\LOCC}(l)$ to $\QMA_{\LOCC}(l-1)$ we had to blow up one of the proof's size only by a quadratic factor, we can repeat the same protocol a constant number of times and still get each proof of polynomial size. In the end, the completeness parameter of the $\QMA$ procedure is the same as the original one for $\QMA_{\LOCC}(l)$, while the soundness is smaller than $s + l \epsilon$, which can be taken to be a constant away from $c$ by choosing $\epsilon$ sufficiently small. To reduce the soundness back to the original value $s$ we then use the standard amplification procedure for $\QMA$ (see e.g.\ \cite{Wat08}), which works in this case since the verification measurement is LOCC \cite{ABDFS08}. 
\end{proof}

\section{Acknowledgments}

We thank Mario Berta, Aram Harrow, David Reeb and Andreas Winter for helpful discussions. We are specially grateful to Ke (Carl) Li and Andreas Winter for pointing out an error in a early version of the paper (namely in the recursion argument for boosting the results from one-way LOCC to general LOCC measurements, as claimed originally). FB is supported by a "Conhecimento Novo" fellowship from the Brazilian agency Fundac\~ao de Amparo a Pesquisa do Estado de Minas Gerais (FAPEMIG). MC is supported by the Swiss National Science Foundation (grant PP00P2-128455) and the German Science Foundation (grants CH 843/1-1 and CH 843/2-1). JY is supported by a grant through the LDRD program of the United States Department of Energy. FB and JY thank the Institute Mittag Leffler, where part of this work was done, for their hospitality.

\bibliographystyle{unsrt}

\end{document}